\documentclass{article}
\usepackage{style}
\usepackage{times}
\usepackage{soul}
\usepackage{url}
\usepackage[hidelinks]{hyperref}
\usepackage[utf8]{inputenc}
\usepackage[small]{caption}
\usepackage{graphicx}
\usepackage{amsmath,amssymb, amsthm}
\usepackage{amsfonts}
\usepackage{booktabs}
\usepackage{algorithm}
\usepackage{algorithmic}
\usepackage[switch]{lineno}

\usepackage{balance} 
\usepackage{mathtools}
\usepackage{svg}

\newtheorem{theorem}{Theorem}
\newtheorem{definition}[theorem]{Definition}
\newtheorem{corollary}[theorem]{Corollary}
\newtheorem{proposition}[theorem]{Proposition}
\newtheorem{lemma}[theorem]{Lemma}
\newtheorem{remark}[theorem]{Remark}
\newtheorem{example}[theorem]{Example}

\usepackage[colorinlistoftodos]{todonotes}
\usepackage{subcaption}

\title{Multi-agent coordination via communication partitions}

\author{
Wei-Chen Lee \and
Alessandro Abate \And
Michael Wooldridge \\
\affiliations
University of Oxford \\
}

\begin{document}

\maketitle

\begin{abstract}
    Coordinating the behaviour of self-interested agents in the presence of multiple Nash equilibria is a major research challenge for multi-agent systems. Pre-game communication between all the players can aid coordination in cases where the Pareto-optimal payoff is unique, but can lead to deadlocks when there are multiple payoffs on the Pareto frontier. We consider a communication partition, where only players within the same coalition can communicate with each other, and they can establish an agreement (a coordinated joint-action) if it is envy-free, credible, and Pareto-optimal. We show that under a natural assumption about symmetry, certain communication partitions can induce social optimal outcomes in singleton congestion games. This game is a reasonable model for a decentralised, anonymous system where players are required to choose from a range of identical resources, and incur costs that are increasing and convex in the total number of players sharing the same resource. The communication partition can be seen as a mechanism for inducing efficient outcomes in this context. 
\end{abstract}

\section{Introduction}

Equilibrium selection in the presence of multiple equilibria has been a long-standing problem in the game theory literature. A natural solution, and one that often takes place in practice, is that players would communicate with each other before their individual actions are chosen. While such communication does not bind players to taking a specific action, it may nonetheless enable them to coordinate on a jointly beneficial outcome. 

\begin{figure}[h] 
    \centering
    \begin{subfigure}[b]{0.14\textwidth}
        \centering        
        \begin{tabular}{|c|c|}
            \hline
            1, 1     & 0, 0   \\
            \hline
            0, 0     & 1, 1   \\
            \hline
        \end{tabular}
        \subcaption{Coordination}
    \end{subfigure}
    \begin{subfigure}[b]{0.16\textwidth}
        \centering        
        \begin{tabular}{|c|c|}
            \hline
            2, 1     & 0, 0   \\
            \hline
            0, 0     & 1, 2   \\
            \hline
        \end{tabular}
        \subcaption{Battle of the sexes}
    \end{subfigure}
    \begin{subfigure}[b]{0.17\textwidth}
        \centering
        \begin{tabular}{|c|c|}
            \hline
            3, 3     & 0, 4   \\
            \hline
            4, 0     & 1, 1   \\
            \hline
        \end{tabular}
        \subcaption{Prisoners' dilemma}
    \end{subfigure}
    \caption{Payoff profiles of well-known games}
    \label{fig:classic_games}
\end{figure}

For instance, in a pure \textit{coordination} game (Figure~\ref{fig:classic_games}a), players can agree to coordinate on one of the two Pareto-optimal outcomes. In other cases players may not readily agree on an efficient outcome if they have different preferences, e.g., the two Nash equilibria in the game \textit{Battle of the Sexes} (Figure~\ref{fig:classic_games}b), or they may fail to honour the agreement if they have beneficial deviations, e.g., agreeing to cooperate in not credible in the \textit{Prisoners' dilemma} (Figure~\ref{fig:classic_games}c). 

In this paper, we consider a general mechanism where players are partitioned into coalitions, and are only able to communicate with those within the same coalition. We show that such a partition can help players reach a more socially efficient outcome, relative to both the cases of no communication and of unrestricted communication where all players are part of a grand coalition.

The idea that coordination among some subset of players can induce more socially efficient outcomes has long been considered in the game theory literature. Consider the 3-player game shown in Figure~\ref{fig:3_person_coordination}, taken from \cite{Aumann74CorrelatedEquilibrium}. The only Nash equilibrium payoff profile is $(1, 1, 1)$, achieved by the row and column players choosing `bottom left', and the matrix player choosing the `left' or `right' matrix. However, players can achieve the socially optimal payoff profile (as measured by the sum of the players' payoff) of $(2, 2, 2)$ if the row and column players can agree on playing `top left' or `bottom right', and the matrix player does not know which agreements is reached and attributes equal probability to both, resulting in a best response of choosing the `middle' matrix. Our proposed mechanism is one way to formalise this idea of partial coordination. 

\begin{figure}[h]
    \centering
    \begin{subfigure}[b]{0.14\textwidth}
        \centering
        \begin{tabular}{|c|c|}
            \hline
            0, 1, 3     & 0, 0, 0   \\
            \hline
            1, 1, 1     & 1, 0, 0   \\
            \hline
        \end{tabular}
    \end{subfigure}
    \hspace{0.1cm}
    \begin{subfigure}[b]{0.14\textwidth}
        \centering
        \begin{tabular}{|c|c|}
            \hline
            2, 2, 2     & 0, 0, 0   \\
            \hline
            2, 2, 0     & 2, 2, 2   \\
            \hline
        \end{tabular}
    \end{subfigure}
    \hspace{0.1cm}    
    \begin{subfigure}[b]{0.14\textwidth}
        \centering
        \begin{tabular}{|c|c|}
            \hline
            0, 1, 0     & 0, 0, 0   \\
            \hline
            1, 1, 1     & 1, 0, 3   \\
            \hline
        \end{tabular}
    \end{subfigure}
    \caption{A 3-player coordination game}
    \label{fig:3_person_coordination}
\end{figure}

We do not explicitly model how players communicate with each other. Rather, we assume that where there is an agreement (a pure joint-action among those within the same coalition) that satisfies envy-freeness, where no player envies another for their assigned action; credibility, where no player has an incentive to unilaterally deviate from the agreement; and Pareto-optimality, where the agreement cannot be improved upon for some players without some other players being worse-off, then such an agreement can be reached. If multiple agreements satisfy these conditions, then communication enables players within the coalition to coordinate on one of the qualified candidates. 

Since players cannot communicate across coalitions, they need to hypothesise about the behaviour of other coalitions in order to assess the expected outcome of their agreement. We propose a simple \textit{symmetry principle}, which says that when faced with uncertainty over symmetric outcomes, players attribute the same probability to each outcome. This principle reflects the fact that players have maximal epistemic uncertainty over these outcomes, and there is no basis for believing one outcome as being more likely than another.

The symmetry principle is applicable in the Singleton Congestion Game (SCG) with identical resources. In such a game, each of $n$ players chooses one among $m$ identical resources, and incur a cost that is increasing and convex in the total number of players choosing the same resource. Naturally, each player wants to choose a resource that is chosen by the least number of other players, but because resources are identical, there is no basis for preferring one resource over another. 

Because each player perceives each resource as being equally likely to be the least chosen resource, each player is indifferent between choosing any of the resources, and can be modelled as choosing each resource with probability $\frac{1}{m}$. From an epistemic perspective, this does not imply that players choose stochastically between resources (although they might), but rather that they each have epistemic uncertainty about the choices of others. From the perspective of an outsider, this epistemic uncertainty about the behaviour of all players is represented by the mixed-strategy Nash equilibrium where all players play each resource with probability $\frac{1}{m}$.

Under the symmetry principle, we show that by dividing player into coalitions, as represented by a communication partition, players within each coalition can reach agreements that are envy-free, credible, Pareto-optimal and socially optimal, in SCGs with identical resources. Importantly, these outcomes are not achievable if players cannot communicate (i.e., in a partition comprising of singleton coalitions), or have unrestricted communication (i.e., in a grand coalition). Thus, the communication partition can be seen as a mechanism that generalises these two standard setups. 

SCG is a simple yet highly relevant game for modelling the anonymous interactions between players choosing congestible resources. While the main focus of the paper is on identical resources, our work applies equally to cases where resources have differential cost functions that are not known to the players. SCG is thus a relevant model in congestion settings where players have epistemic uncertainty not only about the behaviour of other players, but also the cost function associated with each resource. 

The partition mechanism offers two perspectives on pre-game communication. Descriptively, it can be used to explain why certain outcomes are observed in practice by modelling the latent communication between subsets of players; Prescriptively, it can offer guidance on how to design communication partitions that induce socially efficient outcomes.

The remainder of this section discusses the related literature. Section 2 formalises the preliminary ideas in this paper, including the SCG, the communication partition, the three conditions for reaching agreements (envy-freeness, credibility, and Pareto-optimality), and the symmetry principle. Section 3 analyses the conditions for a communication partition to result in envy-free, credible, and Pareto-optimal agreements to be reached within each coalition, and whether collectively such agreements correspond to socially efficient outcomes. Section 4 discusses the extent to which our findings generalises to a broader classes of congestion games with differential resource costs. Section 5 concludes this paper.

\subsection{Related literature}

To our knowledge, the use of a communication partition to augment strategic form games has not been previously studied. This work is nonetheless built on many long established ideas in the game theory literature. 

The idea of representing the communication network between players with a partition was first introduced in \cite{Myerson77Graph-theoretic}. This work focuses on cooperative games with transferable utility, and the role of communication as a means of negotiating for shares of the value of the coalition. In contrast to \cite{Myerson77Graph-theoretic}, our work relates to non-cooperative games with non-transferable utility, where the role of communication is to coordinate actions.

Our work can also be seen as an alternative mechanism to the well-known concept of \textit{correlated equilibrium} developed in \cite{Aumann74CorrelatedEquilibrium,Aumann87CorrelatedEquilibrium}, which relies on a trusted third-party mediator for coordination. In the game shown in Figure \ref{fig:3_person_coordination}, such a device can induce the socially optimal outcome by correlating the actions of the row and column players (either at `top left' or `bottom right') with equal probability, while the probability distribution over these signals is known to the matrix player. 

Since a trusted mediator is not always available, an alternative mechanism relying on `cheap-talk' communication was developed in \cite{Crawford1982,Farrell1996}, where players communicate directly with one another in a non-binding fashion. This led to a rich literature exploring how correlated equilibria (or related concepts) can be implemented in a cheap-talk setting where all players can communicate bilaterally and broadcast messages to all others (e.g., \cite{Barany92CompleteInfoCheaptalk,Heller2012Communication,Forges90CheaptalkIncompleteInfo,BenPorath2003CheaptalkIncompleteInfo,Gerardi2004UnmediatedCommunication,Abraham06RobustSecretSharing,Heller2010Minority-proofCheap-talkProtocol,Abraham2018AsynchronousCheapTalk}. This literature is concerned with implementing a mediated equilibrium concept using a fixed communication network; in contrast, our work concerns how varying such a communication network could induce different outcomes. 

In terms of (atomic) congestion games, the seminal work of \cite{Rosenthal1973CongestionGames}, and later \cite{Monderer1996PotentialGames}, established important and related properties of the game, such as the finite improvement property, existence of pure-strategy Nash equilibria, and the correspondence between congestion games and potential games. Further works on the `price of anarchy' (PoA) of atomic congestion games (e.g., \cite{Koutsoupias2009WorstCaseEquilibria,Christodoulou05PoACongestionGames,Christodoulou2005CorrelatedEquilibriaPoA,Awerbuch2005PoA,Aland2011PoA,Bhawalkar2014PoA}) study the gap between the worse equilibrium outcome and the socially optimal outcome. The SCG that we study is a specific case of the model in \cite{Koutsoupias2009WorstCaseEquilibria} with unweighted players. In this context, our work can be seen as introducing a mechanism for removing poor equilibria and improving the PoA of such games.

A related piece of work \cite{Fotakis2006CoalitionCongestionGames} studies the efficiency of a congestion game of `coalitional' players, where a coalition (i.e., a subset of players) can act as a single player to minimise the maximum cost incurred by any player inside the coalition. This differs from our setup where players within a coalition are self-interested, and that any agreement reached by players within a coalition is non-binding.

On symmetric games, \cite{Nash1951} established the result that every finite symmetric game has a symmetric mixed-strategy Nash equilibrium, which we consider to be the baseline outcome in the absence of communication as a natural consequence of the symmetry principle. We also rely on the known concept of epistemic uncertainty \cite{Aumann1995EpistemicConditions,Aumann2016Epistemics} to model a player's belief about the behaviour of other players. Our work utilises these concepts to develop a model of player behaviour in the presence of a communication partition.

\section{Preliminaries}

\subsection{Singleton Congestion Games}

A Singleton Congestion Game (SCG) is a game with $n$ players, $N$, and $m$ resources, $M$. Players simultaneously choose one of the resources, so $M$ also represents the action set of each player. Each player incurs a cost $f: \mathbb{N} \to \mathbb{Q}$ that is increasing and convex in the total number of players that choose the same resource. That is, given some pure-strategy profile (also referred to as an outcome) $a \coloneqq (a_1, ..., a_n) \in M^n$, the cost to player $i$ is $c_i(a) \coloneqq f(n_{a_i}(a))$, where $n_x(a)$ denotes the number of players that choose resource $x$ in $a$, and $f' > 0$, $f'' \geq 0$ characterise a (strictly) increasing and (weakly) convex function. The game can be specified by $G = \langle n, m, f \rangle$. 

Following \cite{Christodoulou05PoACongestionGames}, we consider two common measures of efficiency: the total cost to all players $\bar c$ (this notation is chosen to avoid potential confusion with other cost-related notations), and the highest cost incurred by any player $\hat c$, where 
\[ \bar c(a) = \sum_{i \in N} f(n_{a_i}(a)) = \sum_{x \in M} n_x(a) f(n_x(a)), \] 
and
\[ \hat c(a) = \max_{i \in N}f(n_{a_i}(a)) = \max_{x \in M}f(n_x(a)) .\]

\subsection{Communication partition}

We use a partition $\pi$ of $N$ to indicate the non-empty, mutually exclusive, and collectively exhaustive subsets of players that can freely communicate with one another. For example, the notation $[12|34|5]$ represents a partition where players $1$ and $2$ can communicate together (i.e. form a communicating coalition), and similarly between players $3$ and $4$, while player $5$ does not communicate with any other player. The partition is common knowledge. The \textit{grand coalition} is the partition where all players are in one coalition (e.g., $[12345]$); the \textit{singleton partition} is the partition where each coalition consists of a single player (e.g., $[1|2|3|4|5]$) and no player can communicate with another player. 

An SCG, $G$, augmented by a communication partition game is specified by $G^{\pi} = \langle n, m, f, \pi \rangle$ and played over two phases. In the communication phase, players within a coalition communicate with each other for an indefinite period of time. Once players are satisfied with their communication, the game proceeds to the action phase, where players choose their actions simultaneously, unbound by any agreement reached during the communication phase. 

We do not explicitly model how players communicate within a coalition. Instead, we assume that players will agree on some pure action-profile $a_C \in M^{|C|}$ and adhere to it in the action phase if $a_C$ is envy-free, credible, and Pareto-optimal, given some shared belief about the behaviour of players outside of their coalition. We unpack each of these assumptions in this subsection. We are not concerned with cases where no such agreement exists, other than to note that more relaxed conditions can be assumed which are beyond the scope of this paper.


\subsubsection{Correlated actions}

Let $C \subseteq N$ be a coalition of players within a partition $\pi$. Because players in $C$ can communicate with each other, they can coordinate on an agreement $a_C \in M^{|C|}$ that could not otherwise be achieved. To better describe the distribution of outcomes that can be achieved through coordination, we require a broader definition than that of a mixed-strategy profile where players are assumed to act independently. 

\begin{definition}[Generalised strategy profile]
    A generalised strategy profile $s \in \Delta (M^n)$ is a probability distribution over all outcomes $M^n$ of the game.
\end{definition}

A generalised strategy profile differs from a mixed-strategy profile $s \in (\Delta M)^n$ in that players can correlate their actions. For example, in a two-player, two-action game with action set $\{T, B\}$ for the row player and action set $\{L, R\}$ for the column player, a generalised strategy-profile can assign $Pr((T, R)) = Pr((B, R) = Pr((B, L)) = 1/3$ and $Pr((T, L)) = 0$, while this distribution cannot be achieve by a mixed strategy profile. However, a generalised strategy profile may not be feasible if some players cannot communicate with each other. To account for the restricted communication imposed by the communication partition $\pi$, we define a feasible strategy profile as follows.

\begin{definition}[Feasible strategy profile]
    A feasible strategy profile $s \in \Delta (M^n)$ given a communication partition $G^\pi$ is a generalised strategy profile where any two players that are not in the same coalition act independently, i.e., for all $C \in \pi$ and for all $i, j$ such that $i \in C$ and $j \not \in C$, $Pr(a_i, a_j) = Pr(a_i) \cdot Pr(a_j)$, where $Pr(a_i)$ denotes the marginal probability that $i$ plays $a_i$ in $s$, and $Pr(a_i, a_i)$ denotes the joint probability that $i$ plays $a_i$ and $j$ plays $a_j$ in $s$.
\end{definition}

This formulation makes explicit the role of the partition, which enables players within a coalition to coordinate on a single, pure action-profile if certain conditions are satisfied. Notice that while players within each coalition may reach a `pure' agreement (a pure action-profile), the resulting feasible strategy profile is stochastic. This is due to the epistemic uncertainty of an outsider about which of the many symmetric agreements has been reached within each coalition, which we discuss below.


\subsubsection{Epistemic uncertainty and the symmetry principle}

To analyse how rational players behave in the absence of communication, we take an epistemic game theory perspective in accordance with \cite{Aumann1995EpistemicConditions}, where each player has beliefs about the beliefs and actions of others. Under the assumption of mutual knowledge of rationality, each player's choice of action is a best-response to their beliefs, and each player knows that others' actions are also best-responses of their beliefs. 

Concretely, each player $i \in N$ has a \textit{conjecture} $\mu^i_{-i}: \Delta (M^{n-1}) \to [0, 1]$, which is a probability distribution over the joint actions of all other players $-i \coloneqq N \setminus \{i\}$. A \textit{rational} player $i$ chooses an action that is a best-response to $\mu^i_{-i}$, i.e., 

\[ \arg \min_{a_i \in M} \mathbb{E}_{\mu^i_{-i}} (c_i(a_i, a_{-i})). \] 

The SCG with identical resources has a specific form of symmetry that we call action-symmetry.\footnote{This is distinct from another form of symmetry, commonly referred to as payoff-symmetry, that is discussed in the literature (e.g. \cite{Plan23Symmetry}).} We say that the game is totally action-symmetric if, for any permutation $\psi: M \to M$, for any $a \in M^n$, and for all $i \in N$, $c_i((a_1, ..., a_n)) = c_i((\psi(a_1), …, \psi(a_n))$. That is, a player's cost remains unchanged if every player's chosen resource is permuted by $\psi$. 

Players are said to have a \textit{common prior} if differences in the players' conjectures $\mu^i_{-i}$ are solely attributable to differences in their information. We make a natural assumption that players abide by what we term the \textit{symmetry principle}, which says that when faced with symmetric uncertainties, a player places equal probability on each outcome. In an action-symmetric game, this assumption leads all players to have the same common prior about the behaviour of each coalition within the partition: if a coalition $C$ can reach an agreement that is envy-free, credible, and Pareto-optimal, then an outside player $i \not \in C$ would hold a conjecture $\mu^i_C$ which places equal probability on all symmetries (i.e. permutations on the labelling of resources) of such agreement. 

For example, consider a singleton coalition consisting of player $j$. By the symmetry principle, $j$ would reason that whatever agreements that can be reached by all remaining players $-j$ in their respective coalitions, such agreements are symmetric about all resource. Thus the expected number of players in $-j$ choosing a resource is the same among all resources, and a rational player $j$ would be indifferent between choosing any of the resources. Whichever resource is chosen by $j$, from the perspective of another player $i \in -j$, the symmetry principle means that $j$ is just as likely to choose any other resources, and thus all players in $-j$ would place probability $1/m$ on $j$ playing any one of the resources.

\begin{remark} \label{rem:symmetry}
The example illustrates that, according to the symmetry principle, all players hold the same common prior in a symmetric game. The action-symmetry of the game arises from the fact that all resources are identical, and thus any relabelling of the resources in an outcome has no impact on a player's payoff. Furthermore, the symmetry principle implies that players across different coalitions always act independently, and thus always play a feasible strategy profile.
\end{remark}


\subsubsection{Expected cost under a common prior} 

We noted in Remark \ref{rem:symmetry} that, under the symmetry principle, every player in the same coalition $C$ has the same symmetric conjecture about the usage of each resource by players outside of $C$, which we denote by $-C \coloneqq N \setminus C$. We show that this effectively means that each coalition of players face, in expectation, an increasing and convex cost function that is identical across resources.

\begin{proposition} \label{prop:subgame}
    The expected cost faced by some player in $C$ who chooses some resource $x$, among a total of $v$ players in $C$ choosing the same resource, can be expressed as $g_\mu(v) \coloneqq \sum_{u=0}^{\infty} \mu(u) \cdot f(u + v)$, where $\mu: \mathbb{N} \to [0, 1]$ is the probability distribution on the number of players $u$ in $-C$ choosing the resource $x$. Then $g_\mu$ is strictly increasing and weakly convex.
\end{proposition}

\begin{proof}
    Since $f$ is strictly increasing, each summand in $g$ is strictly increasing, and therefore $g$ is strictly increasing. For convexity, note that for any $v \in \mathbb{N}$, $g_\mu(v+2) + g_\mu(v) - 2 g_\mu(v+1)
    = \sum_{u = 0}^{\infty} \mu(u) [f(u+v+2) + f(u+v) - 2f(u+v+1)]$. The convexity of $f$ implies that each summand is non-negative, and thus $g_\mu$ is also convex.
    
\end{proof}

\begin{remark} \label{rem:subgames}
Under the symmetry principle, players in $C$ face the same $g_\mu$ for all resources in $M$, and thus we can treat $C$ as a set of players playing a `subgame' of the augmented SCG with a cost function $g_\mu$ that is also increasing and convex. Viewed in this light, the communication partition effectively decomposes an SCG into smaller subgames. This `veil of ignorance' through action-symmetry allows each coalition of players to simply consider their own agreement, without concern for the agreement reached by the other coalitions.
\end{remark}

\begin{corollary} [Coalition as subgame] \label{cor:subgame}
    If players within a coalition $C$ have a common and symmetric conjecture about the choices of others outside of the coalition, then they are equivalently playing an augmented SCG $\langle |C|, m, g, \pi^G \rangle$  where $g$ is increasing and convex, and $\pi^G$ is the grand coalition. 
\end{corollary}


\subsubsection{Envy-freeness, credibility, and Pareto-optimality}

We have so far considered the `hard' constraints that a communication partition imposes on players' ability to coordinate. We now consider the `soft' constraints that incentivise players within a coalition to agree on a certain joint action and adhere to it. We term these conditions envy-freeness, credibility, and Pareto-optimality.

\begin{definition} \label{def:credible_envy-free_Pareto}
    Consider a subset of players $C \subseteq N$ and a common prior $\mu$ about the behaviour of players outside of $C$. 

    An agreement $a_C \in M^{|C|}$ is \textit{envy-free} under $\mu$ if all players incur the same expected cost, i.e. for all $i, j \in C$, \[ \mathbb{E}_{\mu}[c_i(a_C)] = \mathbb{E}_{\mu}[c_j(a_C)] \]

    An agreement $a_C \in M^{|C|}$ is \textit{credible} under $\mu$ if no player can benefit from a unilateral deviation, i.e., for all $i \in C$, $x \in M$,    
    \[ \mathbb{E}_{\mu}[f(n_{a_i}(a_C))] \leq \mathbb{E}_{\mu}[f(n_x(a_C) + 1)] \]

    An agreement is \textit{Pareto-optimal} under $\mu$ if no player can be better-off without some player being worse-off. I.e. $\nexists a'_C \in M^{|C|}$ such that for all $i \in C$,
    \[ \mathbb{E}_{\mu}[c_i(a'_C)] \leq \mathbb{E}_{\mu}[c_i(a_C)] \] 
    and for some $j \in C$,
    \[ \mathbb{E}_{\mu}[c_j(a'_C)] < \mathbb{E}_{\mu}[c_j(a_C)]. \] 
\end{definition}

Credibility is analogous to the concept of a Nash equilibrium for the full set of players $N$. It is a useful concept here to ensure that players adhere to the agreement, and relates to the concept of self-enforcing communication in the cheap-talk literature (e.g., \cite{Farrell1996}). Envy-freeness ensures that an agreement is reached without the possibility of a deadlock where players endlessly negotiate over which agreement $a_C$ to coordinate on. In effect, under envy-freeness all communicating players incur the same expected cost from the agreement. And finally, Pareto-optimality ensures that players do not agree on some sub-optimal joint-action when there are mutually beneficial alternatives.

\begin{example}
    Consider an augmented SCG with 3 players, 2 resources, and players form the grand coalition. There are two classes of agreements in this game: one where all players choose the same resource, and another where two players choose one resource, and one player chooses the other resource. The former is envy-free, but not credible nor Pareto-optimal (all players want to move to the unused resource); the latter is credible and Pareto-optimal, but not envy-free (players sharing the first resource envy the lone player choosing the second resource). 
\end{example}

In the SCG, these concepts are related in agreements where as many resources as possible are used, which we term \textit{covering} agreements.

\begin{definition} [Covering agreements]
    An agreement $a_C$ is \textit{covering} if as many resources are used as possible. I.e., if $|C| \geq m$, then $n_x(a_C) \geq 1$ for all $x \in M$; if $|C| < m$, then $n_x(a_C) \leq 1$ for all $x \in M$.
\end{definition}

\begin{proposition} \label{prop:covering_agreements}
    An agreement $a_C$ is covering if and only if it is Pareto-optimal; a covering agreement that is envy-free is credible; and a credible agreement is a covering agreement.
\end{proposition}

\begin{proof}
    (Covering $\implies$ Pareto-optimal) Any changes to a covering agreement necessarily involves moving some players away from some resource $x$, and adding them to another resource $y$, thereby increasing the cost to players that already use resource $y$. So there can be no Pareto-improvement on a covering agreement. (Pareto-optimal $\implies$ covering) If an agreement is not covering, then some player $i$ can move from a shared resource $x$ to an unused resource $y$. This lowers the costs of player $i$ and all the other players that choose $x$, while the costs of all remaining players are unchanged. A non-covering agreement is therefore not Pareto-optimal. (Covering $\&$ envy-free $\implies$ credible) Consider a covering and non-credible agreement, so that there is a beneficial deviation by some player $i$ from resource $x$ to resource $y$, i.e. $n_x(a_C) > n_y(a_C) + 1$. Because the agreement is covering, there is some player $j$ that chooses $y$ in $a_C$, so $i$ must envy $j$ and the agreement is not envy-free. (Credible $\implies$ covering) A non-covering agreement is clearly not credible, as any players sharing a resource would deviate to an unused one.
\end{proof}

This result is shown as a Venn diagram in Figure \ref{fig:relations}, and simplifies our analysis in the next section: if we want to find an agreement that is Pareto-optimal, envy-free and credible, we need only look for covering agreements that are envy-free.  

\begin{figure}[t]
    \centering
    \includegraphics[width=0.5\linewidth]{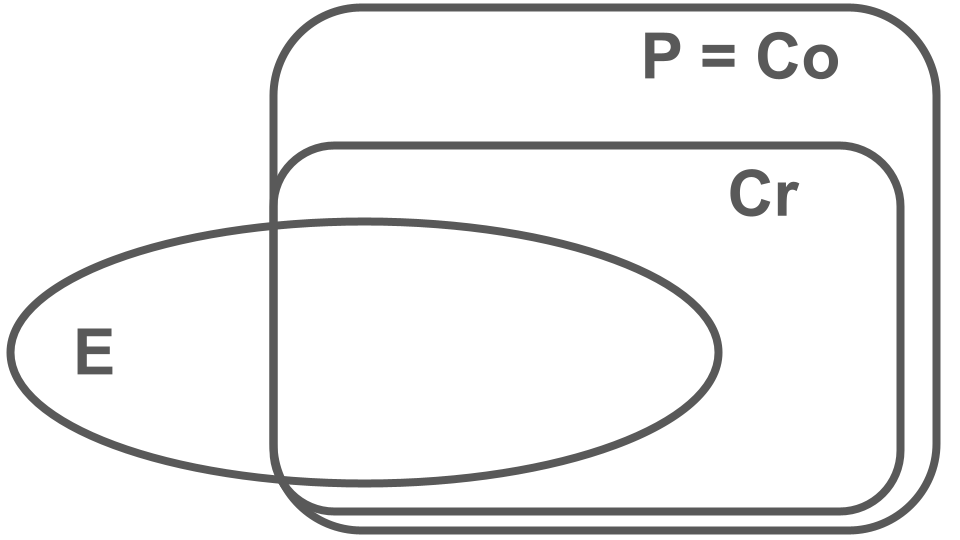}
    \caption{Venn diagram of agreements that are Pareto-optimal (P), covering (Co), credible (Cr), and envy-free (E).}
    \label{fig:relations}
\end{figure}


\section{Efficiency of augmented SCG} \label{sec:analysis}

In this section, we first establish that an outcome that is \textit{evenly distributed} is $\bar c$-optimal and $\hat c$-optimal, and then show that a \textit{balanced} partition induces outcomes that are envy-free, credible, Pareto-optimal, and evenly distributed, and therefore $\bar c$-optimal and $\hat c$-optimal.

\subsection{Optimal outcomes}

We consider an agreement where players choose among resources `as evenly as possible', and show that it is socially optimal with respect to $\bar c$ and $\hat c$.

\begin{definition}[Even distribution]
    A joint action $a_S \in M^{|S|}$ of a subset of players $S \subseteq N$ is \textit{evenly distributed} if for all $x, y \in M$, we have $ |n_x(a_S) - n_y(a_S)| \leq 1$. 
\end{definition} 

The requirement that the differences in the number of players choosing different resource is at most $1$ means that players are spread as evenly as possible among all resources.

\begin{lemma}[Efficient outcomes] \label{lem:efficient_outcomes}
    An outcome $a \in M^n$ is $\bar c$-optimal if and only if $a$ is evenly distributed. An outcome $a \in M^n$ is $\hat c$-optimal if $a$ is evenly distributed.
\end{lemma}

\begin{proof}
    ($\Rightarrow$, $\bar c$-optimality) Consider an outcome $a$ that is not evenly distributed, i.e., there exists some $x, y \in M$ such that $n_x(a) - n_y(a) \geq 2$ (for brevity we drop the reference to $a$ in the remainder of this proof). For $\bar c$, the costs associated with these resources are $n_x f(n_x) + n_y f(n_y)$. If one of the players choosing $x$ changes their choice to $y$, then the costs associated with these resources become $(n_x - 1) f(n_x -1) + (n_y + 1) f(n_y + 1)$. The effect of this change can be expressed as  
    \begin{align*}
        &(n_x - 1)[f(n_x) - f(n_x -1)] - (n_y + 1)[f(n_y + 1) - f(n_y)]\\
        &+ f(n_x)-f(n_y)
    \end{align*}
    
    By assumption, $n_x - 1 \geq n_y + 1$, and by the convexity of $f$, $f(n_x) - f(n_x - 1) \geq f(n_y+1) - f(n_y)$, and thus the sum of the first two terms is non-negative. Since $f$ is increasing, $f(n_x) > f(n_y)$, and the sum of the final two terms is positive. The change therefore decreases $\bar c$, so $a$ cannot be $\bar c$-optimal.

    ($\Leftarrow$, $\bar c$-optimality \& $\hat c$-optimality) Notice that all outcomes where players are evenly distributed have the same social costs $\bar c = l f(y+1) + (n-l) f(y)$ and $\hat c = f(y+1)$, where $y = \lfloor \frac{n}{m} \rfloor$, and $l = n \bmod m$. Since these are the optimal social costs, all evenly distributed outcomes are $\bar c$-optimal and $\hat c$-optimal.
\end{proof}

\begin{figure}[t] 
    \begin{subfigure}[t]{0.21\textwidth}
        \centering
        \includegraphics[scale = 0.3]{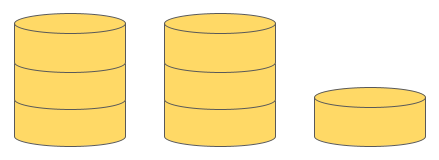}
        \caption{An uneven yet \\$\hat c$-optimal outcome}
    \end{subfigure}%
    \hfill
    \begin{subfigure}[t]{0.26\textwidth}
        \centering
        \includegraphics[scale = 0.3]{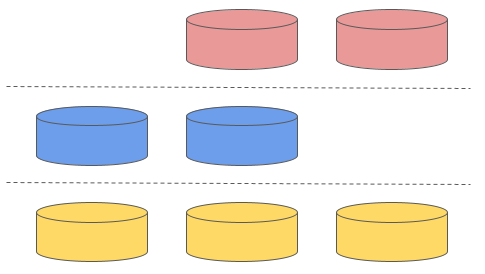}
        \caption{An unbalanced partition that always induces $\hat c$-optimal outcomes}
    \end{subfigure}%
    \caption{Examples of 7-player 3-resource outcomes. Players are represented as coins and resources as columns. A partition is represented by dotted lines and coloured coalitions.}
    \label{fig:unbalanced examples}
\end{figure}

Note that a $\hat c$-optimal outcome need not be evenly distributed. See for example Figure~\ref{fig:unbalanced examples}a.

\subsection{Efficient partition}

We now show that under the symmetry principle, efficient outcomes can be induced by partitioning players in a balanced way.

\begin{definition} [Balanced partition]
    A partition $\pi$ on $N$ is \textit{balanced} if at most one coalition has size such that $|C| < m$, with the remaining coalitions having sizes such that $m \mid |C|$ ($|C|$ is divisible by $m$).
\end{definition}

In other words, a balanced partition is the result of partitioning players into coalitions of sizes that are divisible by $m$, with at most one coalition of `remainders' that is smaller than $m$. Naturally, if the total number of players is divisible by the number of resources, i.e., $m \mid n$, then a partition $\pi$ on $N$ is balanced if and only if the size of every coalition is divisible by $m$.

\begin{definition} [$\pi$-induced partition equilibrium]
    A partition $\pi$ \textit{induces} a feasible strategy profile $s \in \Delta (M^n)$ if, under the symmetry principle, every outcome $a$ in $s$ with positive probability is composed of credible, envy-free and Pareto-optimal agreements $a_C$ for all $C \in \pi$. In this case, we call $s$ a \textit{$\pi$-induced partition equilibrium}.
\end{definition}

The idea of an induced strategy profile is that, given a partition, each coalition of players can reach an agreement that is envy-free, credible and Pareto-optimal, which is expressed as the induced strategy profile when aggregated across all coalitions. Not all partitions induce a feasible strategy profile, since the partition may contain a coalition where there is no agreement that satisfy envy-freeness, credibility and Pareto-optimality (see Lemma \ref{lem:no_agreement_coalition} later).

The definitions of efficiency easily extends from outcomes to feasible strategy profiles: $s \in \Delta (M^n)$ is $\bar c$ (or $\hat c$)-optimal if all its outcomes with positive probability is $\bar c$ (or $\hat c$)-optimal. Therefore a feasible strategy profile is $\bar c$-optimal and $\hat c$-optimal if all the outcomes that occur with positive probability are evenly distributed, since every outcome that is evenly distributed is $\bar c$-optimal and $\hat c$-optimal.

We now establish, for each type of coalition within a balanced partition (i.e., $m \mid |C|$, $|C| < m$, or neither), the conditions for an agreement to be envy-free, credible, and Pareto-optimal. 

\begin{lemma} \label{lem:full_coalition}
    Under the symmetry principle, an agreement $a_C$ reached by a coalition of players $C \in \pi$ of size divisible by $m$ (i.e., $m \mid |C|$) is envy-free, credible, and Pareto-optimal if and only if each resource is chosen by exactly $\frac{|C|}{m}$ players.
\end{lemma}

\begin{proof}
    ($\Rightarrow$) From Corollary \ref{cor:subgame}, players in $C$ need only consider the actions of other players within the coalition. Suppose an agreement $a_C$ is such that some resource $x$ is chosen by more players than some other resource $y$, i.e., $n_x(a_C) > n_y(a_C)$. Then $f(n_x(a_C)) > f(n_y(a_C)))$, and a player choosing $x$ would envy a player choosing $y$, so $a_C$ is not envy-free. ($\Leftarrow$) Suppose an agreement $a_C$ is such that each resource is chosen by the same number of players. It is clear that such an agreement is envy-free (all players incur the same cost) and covering, and therefore by Proposition \ref{prop:covering_agreements} it is also credible and Pareto-optimal.
\end{proof}

\begin{lemma} \label{lem:incomplete_coalition}
    Under the symmetry principle, an agreement $a_C$ among a coalition of players $C \in \pi$ of size $|C| < m$ is envy-free, credible, and Pareto-optimal if and only if no resource is chosen by more than $1$ player.
\end{lemma}

\begin{proof}
    ($\Rightarrow$) From Corollary \ref{cor:subgame}, players in $C$ need only consider the actions of other players within the coalition. Suppose an agreement $a_C$ is such that some resource $x$ is chosen by more than 1 player, and therefore there is some resource $y$ that is not chosen by any player. Then the agreement is not covering, and by Proposition \ref{prop:covering_agreements} not Pareto-optimal. ($\Leftarrow$) Suppose an agreement $a_C$ is such that each resource is chosen by at most 1 player. It is clear that such an agreement is envy-free (all players incur the same cost $f(1)$) and covering, and therefore by Proposition \ref{prop:covering_agreements} it is also credible and Pareto-optimal.
\end{proof}

\begin{lemma} \label{lem:no_agreement_coalition}
    Under the symmetry principle, a coalition of players $C \in \pi$ of size $|C| > m$ and $m \nmid |C| $ cannot reach an envy-free, credible, and Pareto-optimal agreement.
\end{lemma}

\begin{proof}
    For an agreement to be Pareto-optimal, it must be covering, and since $|C| > m$, all resources are used. Because $m \nmid |C|$, some resources must be shared by more players than others, and thus the envy-free condition cannot be satisfied.
\end{proof}

From Lemmas \ref{lem:full_coalition} and \ref{lem:incomplete_coalition}, we can see that a balanced partition $\pi$ induces a very particular feasible strategy profile, where each coalition with $m | |C|$ reach an agreement where all the resources are chosen by the same number of players, and the coalition with $|C| < m$ (if there is one such coalition) reach a covering agreement. We now show that this induced partition equilibrium is always efficient. Moreover, due to Lemma \ref{lem:no_agreement_coalition}, an induced partition equilibrium is $\bar c$-efficient only if the partition which induces it is balanced.

\begin{theorem} [Optimal partition]
    A $\pi$-induced partition equilibrium is $\bar c$-optimal if and only if $\pi$ is balanced. A $\pi$-induced partition equilibrium is $\hat c$-optimal if it is balanced.
\end{theorem}

\begin{proof}
    ($\Rightarrow$, $\bar c$-optimality) 
    By Lemmas\ref{lem:full_coalition}, \ref{lem:incomplete_coalition} and \ref{lem:no_agreement_coalition}, only partitions where for each $C \in \pi$, $m | |C|$ or $|C| < m$ can induce a partition equilibrium, so we need only consider such partitions. For such a partition to be unbalanced, there must exist distinct $C, D \in \pi$ such that $|C| < m$ and $|D| < m$, and $\pi$ induces agreements in $C$ and $D$ that are evenly distributed. If $|C| + |D| \leq m$, then there is a positive probability (since players cannot coordinate across coalitions) that some resource $m$ is chosen by two players, one from each coalition. If $|C| + |D| > m$, then there is a positive probability that some resource $m$ is chosen by none of the players. In either case, there is a positive probability that the resulting outcome is not evenly distributed, and therefore the induced partition equilibrium is not $\bar c$-optimal.

    ($\Leftarrow$, $\bar c$-optimality and $\hat c$-optimality) If $m | n$, then a balanced partition consists of all coalitions of sizes divisible by $m$, so by Lemma \ref{lem:full_coalition} each coalition reaches an evenly distributed agreement, and all possible outcomes are evenly distributed, and therefore the induced partition equilibrium is $\bar c$-optimal and $\hat c$-optimal. If $m \nmid n$, then by Lemma \ref{lem:incomplete_coalition} the coalition $C$ where $|C| < m$ reaches an agreement where at most $1$ player chooses each resource. Combined with agreements reached by all other coalitions, all possible outcomes are again evenly distributed (since the difference in the number of players choosing each resource is at most $1$), and therefore the induced partition equilibrium is $\bar c$-optimal and $\hat c$-optimal.
\end{proof}

Notice that a partition need not be balanced to induce a $\hat c$-optimal generalised strategy profile. For an example see Figure \ref{fig:unbalanced examples}b.


\section{Differential cost functions}

A natural extension to consider is whether our results generalise to broader classes of SCG. In this section we briefly discuss the extent to which our analysis can be generalised to resources with differential cost functions. Another interesting extension concerning weighted players is too broad to cover in this paper, and we defer such discussion to a future paper.

When resources have different costs $f_x$ for each $x \in M$, and that these costs are known to the players, we face two further issues concerning whether there exists a communication partition that induces efficient outcomes.

First, envy-freeness becomes an even stronger condition to satisfy. Consider a simple example of a game of 2 players $\{1, 2\}$ and 2 resources $\{x, y\}$, where $f_x(n) = 1.1 \cdot f_y(n)$, that is, resource $x$ is a little more costly than $y$ given the same number of users $n$. We first consider the case where the partition is the grand coalition. Clearly, the $\bar c$-optimal and $\hat c$-optimal outcome is for the players to choose different resources, and such an agreement would also be credible and Pareto-optimal. However, such an agreement could not be envy-free, since the player that is allocated $x$ would envy the player that is allocated the cheaper resource $y$, and thus players might reach a deadlock. Alternatively, both players choosing the same resource $x$ or $y$ is clearly envy-free, but not credible nor Pareto-optimal. 

Second, the action-symmetry of the game is broken, and we can no longer rely on the symmetry principle to provide a common, symmetric prior. In the above example, if the partition is a singleton partition (i.e., players cannot communicate), then it is not clear how each player should model the behaviour of the other. Additional assumptions about how players behave (or more precisely, how players believe others behave) are required in order analyse this problem.

Nevertheless, both issues fall away if we assume that all players are aware of the distribution of the cost functions among resources, but not the cost associated with each resource. As we showed in Proposition \ref{prop:subgame}, a convex combination of increasing and convex functions is also increasing and convex, and the expected cost function is again identical across resources. Thus players will behave as if resources are identical. Our analysis on the efficient partition in SCG is therefore applicable to games with differential cost functions where only the distribution is known.

\section{Conclusion}

Player communication is a natural consideration in congestion games, where players are self-interested but wish to avoid poorly coordinated outcomes. However, allowing all players to communicate with each other may be practically infeasible, nor satisfactory (from an envy-free perspective) for the players involved. The use of an appropriately designed communication partition can simplify the coordination problem for all players concerned and avoid envy-induced deadlocks.

We believe this mechanism can be applied to large, decentralised resource allocation problems (e.g., in an online anonymous setting), which we model using the SCG in this work. Some of the natural extensions of this problem, such as having resources of known and differential cost functions, and players with known and differential weights, are interesting computational problems that warrant further studies.

\bibliographystyle{named}
\bibliography{references}

\end{document}